\documentclass[smallextended]{svjour3}       % onecolumn (second format)
\smartqed  % flush right qed marks, e.g. at end of proof
\usepackage{graphicx}
\usepackage{amsfonts,amssymb,amscd,amsmath,enumerate,verbatim,calc}
\usepackage[colorlinks,citecolor=blue]{hyperref}
\usepackage{color}

\spnewtheorem{result}{Result}{\bf}{\it}

\DeclareMathOperator*{\Ric}{Ric}

\begin{document}

\title{Einstein equations with cosmological constant in Super Space-Time}
%placed here. General acknowledgments should be placed at the end of the %article.}
%\subtitle{On warped product spaces with cosmological constant}

\titlerunning{  }        % if too long for running head

\author{F. Gholami,  F. Darabi, M.  Mohammadi, S. Varsaie, M. Roshandel}

\authorrunning{F. Gholami \and F. Darabi
        } % if too long for running head
\institute{F. Gholami (Corresponding author)\at
           Institute for Advanced Studies
in Basic Sciences, Zanjan, Iran\\
              \email{gholamifatemeh21@gmail.com}
               \and
                F. Darabi  \at
           Department of Physics, Azarbaijan Shahid Madani University, Tabriz, Iran\\
              \email{f.darabi@azaruniv.ac.ir}
\and
 M.  Mohammadi \at
 Institute for Advanced Studies
in Basic Sciences, Zanjan, Iran\\
\email{moh.mohamady@iasbs.ac.ir} \and
 S. Varsaie \at
 Institute for Advanced Studies
in Basic Sciences, Zanjan, Iran\\
\email{varsaie@iasbs.ac.ir} \and
 M.  Roshandel \at
 Institute for Advanced Studies
in Basic Sciences, Zanjan, Iran\\
\email{m.roshandel@iasbs.ac.ir}}

\date{Received: date / Accepted: date}

\maketitle
\begin{abstract}
We introduce a new kind of super warped product spaces
$\bar{M}_{_{(I)}}=\textbf{I}^{1|0}\times_f M^{m|n}$,
$\bar{M}_{_{(II)}}=\textbf{I}^{0|1}\times_{f} M^{m|n}$, and
$\bar{M}_{_{(III)}}=\textbf{I}^{1|1}\times_{f} M^{m|n}$, where
$M^{m|n}$ is a supermanifold of dimension $m|n$,
$\textbf{I}^{\delta|\delta'}$ is standard superdomain with
$\textbf{I}=(0,1)$ and $\delta,\delta' \in \{0,1\}$, subject to the
warp functions $f(t)$, $f(\bar t)$, and $f(t, \bar t)$,
respectively. In each super warped product space,
$\bar{M}_{_{(I)}}$, $\bar{M}_{_{(II)}}$, and $\bar{M}_{_{(III)}}$,
it is shown that Einstein equations
$\bar{G}_{AB}=-\bar{\Lambda}\bar{g}_{AB}$, with cosmological term
$\bar{\Lambda}$ are reducible to the Einstein equations
$G_{\alpha\beta} = -\Lambda g_{\alpha\beta}$ on the super space
$M^{m|n}$ with cosmological term ${\Lambda}$, where $\bar{\Lambda}$
and ${\Lambda}$ are functions of $f(t)$, $f(\bar t)$, and $f(t, \bar
t)$, as well as ($m$, $n$). This dependence points to the origin of
cosmological terms which turn out to be within the warped structure
of the super space-time.  By using the Generalized Robertson-Walker
space-time, as a super space-time, and demanding for constancy of
$\bar{\Lambda}$ we can determine the warp functions and $\Lambda$
which result in finding the solutions for Einstein equations
$\bar{G}_{AB}=-\bar{\Lambda}\bar{g}_{AB}$ and $G_{\alpha\beta} =
-\Lambda g_{\alpha\beta}$. We have discussed the cosmological solutions,  for each kind of super warped product space, in the special case of $M^{3|0}$.
%\keywords{ Reduction \and Einstein equations \and Warped product }

\end{abstract}

\section{Introduction}
Bishop and O'Neill were the pioneering in introducing the warped
product spaces to construct Riemannian manifolds having negative
curvature \cite{bishop.oneill}. The warped product spaces were
widely studied in the context of general relativity theory to
construct new metrics with interesting geometrical and physical
properties. In this regard, some solutions of Einstein equations,
like generalized Friedmann-Robertson-Walker metric, generalized
Schwarzschild black hole metric, generalized Reissner-Nordstrom
black hole metric, generalized (2 + 1)-dimensional
Banados-Teitelboim-Zanelli (BTZ) black hole metric, generalized
(2+1)-dimensional de Sitter black hole metric, generalized standard
static metric, 5-dimensional Randall-Sundrum and
Dvali-Gabadadze-Porrati metrics, generalized twisted product
structure and special base conformal warped product structure were
shown by the authors to be expressed in terms of multi warped
products in Lorentzian geometry \cite{FBG}-\cite{GDB2}.

By considering the importance of warped products in geometry and
physics, we present a meaningful generalization of warped product
which is called {\it super warped product}. This generalization is
novel in its kind and has not been defined before. We study Einstein
equations with cosmological term in this super warped product space
and show how to reduce them to the Einstein equations with
cosmological term in the lower dimensional space. It is shown that
the cosmological terms are defined in terms of the warp function in
the super warped product space and the dimensions of supermanifold.
In explicit words, we consider three types of super warped product
metric of the types $\bar{M}_{_{(I)}}=\textbf{I}^{1|0}\times_f
M^{m|n}$, $\bar{M}_{_{(II)}}=\textbf{I}^{0|1}\times_{f} M^{m|n}$,
and $\bar{M}_{_{(III)}}=\textbf{I}^{1|1}\times_{f} M^{m|n}$ where
$M^{m|n}$ is a supermanifold of dimension $m|n$,
$\textbf{I}^{\delta|\delta'}$ is standard superdomain with
$\textbf{I}=(0,1)$ and $\delta,\delta' \in \{0,1\}$, subject to the
warp functions $f(t)$, $f(\bar t)$, and $f(t, \bar t)$,
respectively. Similar to the previous examples, we show that in each
super warped product space, $\bar{M}_{_{(I)}}$, $\bar{M}_{_{(II)}}$,
and $\bar{M}_{_{(III)}}$, the Einstein equations
$\bar{G}_{AB}=-\bar{\Lambda}\bar{g}_{AB}$, with cosmological term
$\bar{\Lambda}$ are reducible to the Einstein equations
$G_{\alpha\beta} = -\Lambda g_{\alpha\beta}$ on the super space $M$
with cosmological term ${\Lambda}$, where $\bar{\Lambda}$ and
${\Lambda}$ are functions of $f(t)$, $f(\bar t)$, and $f(t, \bar
t)$, as well as ($m$, $n$). Moreover, by using the Generalized
Robertson-Walker space-times and demanding for constancy of
$\bar{\Lambda}$ we can determine the warp functions and $\Lambda$
which result in finding the solutions for Einstein equations
$\bar{G}_{AB}=-\bar{\Lambda}\bar{g}_{AB}$ and $G_{\alpha\beta} =
-\Lambda g_{\alpha\beta}$. For each kind of super warped product space, in
the special case of $M^{3|0}$, we have discussed the cosmological solutions.
%\keywords{ Reduction \and Einstein equations \and Warped product }

\section{Preliminaries}
In this section we need to state some definitions of super manifolds
and  super Riemannian  metric.

%%%%%%%%%%%%%%%%
 \begin{defi}
By a super ringed space, we mean a pair $(X, \mathcal O_X )$ where
$X$ is a topological space and $\mathcal O_X$ is a sheaf of
supercommutative $\mathbb Z_2$-graded rings on $X$. A morphism
between
%two super ringed spaces
$(X, \mathcal O_X)$ and $(Y, \mathcal O_Y)$ is a pair
$\psi:=(\bar{\psi},\psi^*)$ such that $\bar{\psi}:X\rightarrow Y$ is
a continuous map and $\psi^*:\mathcal O_Y\rightarrow
\bar{\psi}_*\mathcal O_X$ is a homomorphism between the sheaves of
supercommutative $\mathbb Z_2$-graded rings.   The super ring space
$ U^{m|n}:=\big(U,C^\infty_{U^m}\otimes \wedge \mathbb R^n\big)$,
$U$ is an open subset of $\mathbb R^m$, is called standard
superdomain where $C^\infty_{U^m}$ is the sheaf of smooth functions
on $U$ and $\wedge \mathbb R^n$ is the exterior algebra of $\mathbb
R^n$.
\end{defi}

\begin{defi}
A supermanifold of dimension $m|n$ is a super ringed space
$(\bar{M},\mathcal O_M)$ that is locally isomorphic to $\mathbb
R^{m|n}$ and $\bar{M}$ is a second countable and Hausdorff
topological space.\\
It can be shown that the stalks of  the structure sheaf of a
supermanifold are local rings.

 A morphism between two supermanifolds $M=(\bar{M},\mathcal O_M)$ and $N=(\bar{N},\mathcal O_N)$ is just a
 morphism say $(\bar{\psi},\psi^*)$
 between two super ringed spaces such that for $y=\bar{\psi}(x)$, $\psi^*:  \mathcal O_{N_y} \rightarrow  \mathcal O_{M_x}$ is  a morphism   of local rings
   . By $J_M(U)$ for open subset $U\subset\bar{M}$,  we mean the set of all nilpotent elements in $\mathcal O_M(U)$. It can be seen that
   the quotient sheaf ${\mathcal O_M}/{J_M}$ is  locally   isomorphic  to the sheaf $C^{\infty}_{\mathbb R^m}$
  . (See \cite{Vara}.)
   Thus $\widetilde{M}:=(\bar{M},{\mathcal O_M}/{J_M})$ is a classical smooth manifold that is called \textit{reduced manifold} associated to $M$.  In addition
   each morphism  of supermanifolds  $\psi:M\rightarrow N$ induces
   a smooth map $\widetilde{\psi}:\widetilde{M}\rightarrow
   \widetilde{N}$.\\
In supergeometry one may define the tangent bundle $\mathcal T_M$
and differential forms $\Omega_M^k$ for each supermanifold $M$. see
\cite{Vara} and \cite{CC-LC-RF} for more details.
\end{defi}
The category of supermanifolds admits products. Let  $M_i (1 < i <
n)$ be spaces in the category. A ringed space  $ M $ together with
("projection") maps  $ P_i : M\rightarrow  M_i$  is called a product
of the  $ M_i $, and is denoted by
$$M=M_1\times...\times M_n,$$ if the following is satisfied: for any
ringed space $N$, the map  $$ f\longmapsto ( P_1\circ f  , ... , P_n
\circ f),$$ from  $Hom(N, M)$ to $\Pi_i Hom(N, M_i)$  is a
bijection. In other words, the morphisms $f$ from $N$ to $ M $ are
identified with n-tuples $ (f_1, ... ,f_n) $ of morphisms such that
$f_i (N \rightarrow  M_i) $for all $i$.\cite{Vara}
\begin{defi}
Let $V=V_0\oplus V_1$ be a super vector space over a field $\mathbb
K.$ By a scalar superproduct on $V,$ we mean a non-degenerate and
supersymmetric even $\mathbb K$-bilinear form $\langle \ldotp ,
\ldotp \rangle:V \times V\rightarrow V$. By the non-degeneracy, we
mean that the mapping $v\mapsto \langle v , \ldotp \rangle$ is an
isomorphism for all $v\in V$ and by  graded symmetric, we mean that
for all homogeneous elements $v,w \in V,$ we have
$$\langle v , w \rangle=(-1)^{|v||w|}\langle w , v \rangle,$$
where $|v|$ is the parity of $v.$
 \end{defi}
 \begin{defi}
 A Riemannian supermetric on a supermanifold $M$ is a graded symmetric even
 non-degenerate $\mathcal O_M$-linear morphism of sheaves
$$g:\mathcal T_M\otimes \mathcal T_M\rightarrow \mathcal O_M,$$
where by the non-degeneracy we mean that the mapping $X\mapsto
g(X,\ldotp)$ is an isomorphism from $\mathcal T_M$ to $\Omega_M^1.$
A supermanifold equipped with a Riemannian supermetric is called a
Riemannian supermanifold.
\end{defi}
\section{ Super warped product}
 In this section, we define super warped product space. To see the
 analogue definition in classical geometry one may refer to
\cite{bishop.oneill}.\\

\begin{defi}
Let $(B, g_B)$ and $(M, g_M)$ be semi-Riemannian supermanifolds of
dimensions $k|l$  and $m|n$ respectively. Let $f\in \mathcal{O}(B)$
be a  superfunction  with $\tilde{f}>0$ where $\tilde{f}$ is a smooth function on $\tilde{B}$ such that, for each $x\in \tilde{B}, f-\tilde{f}(x)$ is not invertible on every neighbourhood of $x$. For more detailed description see [6]. The super warped product is the
product supermanifold $\mathcal M = B \times M$ together with the
supermetric defined by $g = \pi_1^*(g_B) \oplus \pi_1^*(f)
\pi_2^*(g_M)$ where $\pi_1$ and $\pi_2$ are the natural projections
of $B\times M$ to $B$ and $M$ respectively. In addition $f$ is
called super warp function and  the supermanifolds $(M, g_M)$ and
$(B, g_B)$ are called fiber and  base supermanifolds of the super
warped product respectively.
\end{defi}
%%%%%%%%%%%%%%%%
\section{  Generalization of  Robertson-Walker Space-times in Super-geometry}
Generalized Robertson-Walker space-times were introduced in 1995 by
Al\'ias, Romero and S\'anchez (see Ref.\cite{LRS,RS}). In this
section, we study  a generalization of  Robertson-Walker space-times
in super-geometry and obtain the Einstein equations with
cosmological constant.

 Let $\bar{M}=\textbf{I}^{\delta|\delta'}\times_{f} M^{m|n}$ be   a super warped product where  $M^{m|n}$ is semi-Riemannian supermanifold  and $\textbf{I}^{\delta|\delta'}$ is standard superdomain with $\textbf{I}=(0,1)$ and $\delta,\delta' \in \{0,1\} $. The supermanifold $\bar{M}$ is
 equipped by
 superlorentzian metric
\begin{equation}
\bar{g}=-\pi_1^*(g_{\textbf{I}^{\delta|\delta'}})+{f}\pi_2^*(g_{M^{m|n}}).
\end{equation}
Then we call $(\bar{M},\bar{g})$ a super semi-Riemannian warped
product . We use the Einstein convention, that is, repeated indices
with one upper index and one lower index denote summation over their
range. If not stated otherwise, throughout the paper we use the
following ranges for indices:  $ i, j, k, ...\in \{1,...,m\}$ and
$\alpha, \beta,
...\in \{m+1,...,m+n\}$ for even and odd indices of $M^{m|n}$ respectively.\\

Now, in the following cases, we study Einstein equations with
cosmological constant in the Generalized Robertson-Walker
space-times.
 \\
\subsection*{\textbf{Case I}}
Let $\bar{M}=\textbf{I}^{1|0}\times_{f} M^{m|n}$, $( x^i,x^\alpha)$
be an even - odd  coordinate system on $M^{m|n}$, $t$ be  a
coordinate on $(0,1)$ and $f\equiv f(t)$. By $|t|,|i|$ and
$|\alpha|$, we mean the parities of $t,x^i$ and $x^\alpha$
respectively. So $|t|=0= |i|$ and
$|\alpha|=1$.\\
Then, we have
\begin{align}\label{wmetr}
&\bar{g}_{tt}=\bar{g}(\partial_ t,\partial _t)=-1
,\nonumber\\
&\bar{g}_{\alpha t}= \bar{g}(\partial_ \alpha,\partial
_t)=\bar{g}_{t \alpha }=\bar{g}_{t i }=0 , \nonumber\\
&\bar{g}_{\alpha\beta}=fg_{\alpha\beta}, \nonumber\\
&\bar{g}_{ij}=fg_{ij},\nonumber\\
& \bar{g}_{i\alpha}=fg_{i\alpha},
\end{align}
where  $\partial_t= \frac{\partial}{\partial t}$, $\partial_i= \frac{\partial}{\partial x_i}$ and  $\partial_\alpha= \frac{\partial}{\partial {x_\alpha}} .$\\
For the next proposition we need the following lemma.
\begin{lemma}
Let $\bar{M}=\textbf{I}^{1|0}\times_{f} M^{m|n}$ be a super
semi-Riemannian warped product. The Christoffel symbol of
$(\bar{M},\bar{g})$ admits the following expression
\begin{align}\label{wGamma1}
\bar{\Gamma}_{IJ}^L
=&\frac{1}{2}\Big[\partial_{_I}\bar{g}_{_{Jt}}-\partial_t\bar{g}_{_{IJ}}+(-1)^{^{|I||J|}}
\partial_{_J}\bar{g}_{_{tI}}\Big]\bar{g}^{tL}\nonumber \\
&+
\frac{1}{2}\Big[\partial_{_I}\bar{g}_{_{Js}}-\partial_s\bar{g}_{_{IJ}}+(-1)^{^{|I||J|}}
\partial_{_J}\bar{g}_{_{sL}}\Big]\bar{g}^{sL}\nonumber\\&
 +
\frac{1}{2}\Big[\partial_{_I}\bar{g}_{_{J\alpha}}-(-1)^{^{(|I|+|J|)}}\partial_\alpha\bar{g}_{_{IJ}}+(-1)^{^{|I|(|J|+1)}}
\partial_{_J}\bar{g}_{_{\alpha I}}\Big]\bar{g}^{\alpha
L},
\end{align}
where  $I,J,K,L$ are arbitrary indices on $\bar M$, i.e. they can
stand for even and odd indices $ t,s,\alpha$.
\end{lemma}
\begin{proof}
Let $N$ be a supermanifold; for Christoffel symbols we have the
following equation
\begin{align}
\bar{\Gamma}_{IJ}^L
=&\frac{1}{2}\Big[\partial_{_I}\bar{g}_{_{JK}}-(-1)^{^{|K|(|I|+|J|)}}\partial_{_K}\bar{g}_{_{IJ}}+(-1)^{^{|I|(|J|+|K|)}}
\partial_{_J}\bar{g}_{_{KI}}\Big]\bar{g}^{KL},
\end{align}
where
$K \in\{t,s,\alpha\}$.
\end{proof}
Simply, one can show that
\begin{align}
&\bar{\Gamma}_{tt}^t = \bar{\Gamma}_{It}^t =\bar{\Gamma}_{tt}^L =0,
~~~~ \bar{\Gamma}_{tI}^L = \bar{\Gamma}_{It}^L
=\frac{f'}{2f}\delta_{IL}, \nonumber\\
&\bar{\Gamma}_{IJ}^t =\frac{1}{2}f'g_{IJ}, ~~~~~\bar{\Gamma}_{IJ}^L
=\Gamma_{IJ}^L,
\end{align}
where $I,J,L\in\{s,\alpha\}$ and $'$ denotes $\partial_t$.
\begin{lemma}
Let $\bar{M}=\textbf{I}^{1|0}\times_f M^{m|n}$ be a super
semi-Riemannian warped product.
\begin{align}\label{wRieman1}
\bar{R}_{IJK}^L=\partial_{_I} \bar{\Gamma}_{JK}^L +&
(-1)^{^{|I|(|J|+|K|)}}\bar{\Gamma}_{JK}^t\bar{\Gamma}_{It}^L +
(-1)^{^{|I|(|J|+|K|)}}\bar{\Gamma}_{JK}^s\bar{\Gamma}_{Is}^L
\nonumber \\ & + (-1)^{^{|I|(|J|+|K| +1)}}\bar{\Gamma}_{JK}^
{\alpha}\bar{\Gamma}_{I\alpha}^L -
(-1)^{^{^{|I||J|}}}\partial_{_J}\bar{\Gamma}_{IK}^L -
(-1)^{^{|J||K|}}\bar{\Gamma}_{IK}^t\bar{\Gamma}_{Jt}^L \nonumber
\\ & -
(-1)^{^{|J||K|}}\bar{\Gamma}_{IK}^s\bar{\Gamma}_{Js}^L-
(-1)^{^{|J|(1+|K|)}}\bar{\Gamma}_{IK}^{\alpha}\bar{\Gamma}_{J\alpha}^L,
\end{align}
where $ \bar{R}$ is the curvature tensor of  $\nabla$. In addition $
[\nabla_{\partial_ I}, \nabla_{\partial_{_J}}]\partial_{_K} =
\bar{R}_{IJK}^L \partial_{_L}$ and $I,J,K,L$ stand for even and odd
indices $ t,i,\alpha$
\end{lemma}
\begin{proof}
Let $\nabla$ be the Levi-Civita connection of the metric $\bar{g}$.
Then one has
$\nabla_{\partial_{_I}}\partial_{_J}=\Gamma_{IJ}^K\partial_{_K}$.
Since the connection $\nabla$  is torsion free we get
\begin{align*}
&\Gamma_{IJ}^L g_{_{LK}} + \Gamma_{IJ}^{\bar{L}}
g_{_{\bar{L}K}}=\frac{1}{2}\left(
\partial_{_I} g_{_{JK}}- (-1)^0 \partial_{_K}
g_{_{IJ}}+(-1)^{^{|I||J|}}\partial_{_J} g_{_{KI}}\right),\\
&\Gamma_{IJ}^L g_{_{L\bar{K}}} + \Gamma_{IJ}^{\bar{L}}
g_{_{\bar{L}\bar{K}}}= \frac{1}{2}\left(
\partial_{_I} g_{_{J\bar{K}}}- (-1)^{^{(|J|+|K|)}}
\partial_{_{\bar{K}}}
g_{_{IJ}}+(-1)^{^{|I|(|J|+1)}}\partial_{_J} g_{_{\bar{K}I}}\right).
\end{align*}
By the equation \eqref{wGamma1}, and a straightforward computation
according the argument in \cite{GO} we will have the desired result.
\end{proof}
\begin{pro}
Let $\bar{M}=\textbf{I}^{1|0}\times_{f} M^{m|n} $ be a super
semi-Riemannian warped product then one has
\begin{align}\label{wRicc1}
&\bar{\Ric}_{tt}= (n-m)(\frac{f''}{2f}-\frac{f'^2}{4f^2}), \nonumber\\
&\bar{\Ric}_{t\alpha}=\bar{\Ric}_{ti}=0, \nonumber \\
&\bar{\Ric}_{\alpha\beta}={\Ric}_{\alpha\beta} +
\left((m-n-2)\frac{f'^2}{4f^2}+\frac{f''}{2f}\right)g_{\alpha\beta}, \nonumber\\
&\bar{\Ric}_{ij}={\Ric}_{ij} +
\left((m-n-2)\frac{f'^2}{4f^2}+\frac{f''}{2f}\right)g_{ij}, \nonumber\\
&\bar{\Ric}_{\alpha i}={\Ric}_{\alpha i} +
\left((m-n-2)\frac{f'^2}{4f^2}+\frac{f''}{2f}\right)g_{\alpha i},
\end{align}
where $ {\Ric}_{\alpha\beta} , {\Ric}_{\alpha i},{\Ric}_{ij} $ are
the Ricci tensors of semi-Riemannian supermanifold   $(M, g)$.
\end{pro}
\begin{proof}
Let $N$ be an arbitrary supermanifold then we have
\begin{align*}
\bar{\Ric}(\partial_{_I},\partial_{_J})= R_{sIJ}^s
-(-1)^{^{(|I|+|J|)}}R_{\alpha IJ}^\alpha,
\end{align*}
where $s, \alpha$ are even and odd indices over $N.$ If
$N:=\bar{M}=\textbf{I}^{1|0}\times_{f} M^{m|n}$ we have
\begin{align}
\bar{\Ric}(\partial_{_I},\partial_{_J})= R_{tIJ}^t + R_{sIJ}^s
-(-1)^{^{(|I|+|J|)}}R_{\alpha IJ}^\alpha,
\end{align}
where $\bar{R} _{IJK}^L$ is curvature tensor as above and $I,J
\in\{t,s,\alpha\}.$

\end{proof}
\begin{pro}
Let $\bar{M}=\textbf{I}^{1|0}\times_{f} M^{m|n} $  be a super
semi-Riemannian warped product. Then the scalar curvature $\bar{S}$
of $(\bar{M},\bar{g})$ admits the following expressions
\begin{equation}\label{wScalar1}
\bar{S}=\frac{S_M}{f}+(m-n)\frac{f''}{f}+(m-n)(m-n-3)\frac{f'^2}{4f^2}.
\end{equation}
\end{pro}
\begin{proof}
We use the following equation
\begin{align}
\bar{S}=\bar{\Ric}_{tt}\bar{g}^{tt}+\bar{\Ric}_{ij}\bar{g}^{ji}-\bar{\Ric}_{\alpha\beta}\bar{g}^{\beta\alpha}.
\end{align}
  Then,  with a straightforward computation
  we will have the desired result.
\end{proof}
\begin{pro}
Let $\bar{G}$ be the Einstein gravitational tensor field of
$(\bar{M},\bar{g})$, then we have the following equations
\begin{align}\label{wEnest1}
\bar{G}_{tt}&=\frac{S_M}{2f}+ (m-n)(m-n-1)\frac{f'^2}{8f^2},\nonumber\\
\bar{G}_{ij}&=G_{ij}+(1-m+n)\frac{f''}{2}g_{ij}+(m-n-2)\frac{f'^2}{4f}g_{ij}-
(m-n)(m-n-3)\frac{f'^2}{8f}g_{ij},\nonumber\\
\bar{G}_{i\alpha}&=G_{i\alpha}+(1-m+n)\frac{f''}{2}g_{i\alpha}+(m-n-2)\frac{f'^2}{4f}g_{i\alpha}-
(m-n)(m-n-3)\frac{f'^2}{8f}g_{i\alpha},\nonumber\\
\bar{G}_{\alpha \beta}&=G_{\alpha
\beta}+(1-m+n)\frac{f''}{2}g_{\alpha
\beta}+(m-n-2)\frac{f'^2}{4f}g_{\alpha \beta}-
(m-n)(m-n-3)\frac{f'^2}{8f}g_{\alpha
\beta},\nonumber\\
\bar{G}_{ti}&=\bar{G}_{t \alpha}=0.
\end{align}
\end{pro}
\begin{proof}
The Einstein gravitational tensor field of $(\bar{M}, \bar{g})$ is
\begin{align}
\bar{G} = \bar{\Ric} - \frac{1}{2}\bar{S}\bar{g},\label{Q_6}.
\end{align}
 By using \eqref{Q_6}, \eqref{wRicc1}, \eqref{wScalar1} and \eqref{wmetr},
 we obtain \eqref{wEnest1}.
\end{proof}
\begin{pro}
Let $(\bar{G}, \bar{\Lambda})$ be the Einstein gravitational tensor
field and cosmological constant of $(\bar{M},\bar{g})$ for the case
I, then we have the following equations
\begin{align}\label{wConst1}
&\bar{\Lambda}(t)=\frac{1}{4}(m-n-1)\left((m+n)\frac{f''}{f}-(m+3n)\frac{f'^2}{2f^2}\right),\\
&G_{\alpha\beta}=\frac{1}{4}(1-m+n)(m+n-2)\left(\frac{f''}{f}-\frac{f'^2}{f^2}\right)fg_{\alpha\beta
}.\label{G-alpha-beta1}
\end{align}
\end{pro}
\begin{proof}

The Einstein equations with cosmological  constant $\bar{\Lambda}$
is
\begin{align}
\bar{G}=-\bar{\Lambda}\bar{g}.\label{Einst-equation}
\end{align}
 By using  \eqref{Einst-equation},\eqref{wmetr}  and \eqref{wEnest1} , we
 obtain
\eqref{wConst1} and \eqref{G-alpha-beta1}.
\end{proof}
The Einstein equations  on $(\bar{M},\bar{g})$ with cosmological
constant $\bar{\Lambda}$ induces the Einstein equations
 on $(M,
g)$, where the cosmological constant $\Lambda$ is given by
\begin{equation}
\Lambda(t)=\frac{1}{4}(m-n-1)(m+n-2)\left(\frac{f''}{f}-\frac{f'^2}{f^2}\right)f(t).
\end{equation}
We call $\bar{\Lambda}$ and ${\Lambda}$ here as
$\bar{\Lambda}_{_{(I)}}$ and ${\Lambda}_{_{(I)}}$, respectively. The
requirement $\bar{\Lambda}_{_{(I)}}={\it Const}$ can determine the
warp function $f(t)$ and $\Lambda_{_{(I)}}(t)$, using which the
Einstein equations
$\bar{G}_{AB}=-\bar{\Lambda}_{_{(I)}}\bar{g}_{AB}$ and
$G_{\alpha\beta}=-{\Lambda}_{_{(I)}}g_{\alpha\beta} $ are solved.

Let us consider $\bar{\Lambda}_{_{(I)}}=L^{-2}={\it Const}$, where
$L$ has dimension of {\it length}. Then, we obtain
\begin{equation}\label{f(t)}
f(t)=\frac{L^2}{4k}\exp{\frac{t}{\sqrt{k}L}},~~~~k=\frac{1}{8}(m+n)(m-n-1).
\end{equation}
Assuming a supermanifold $M^{3|0}$, the warp function in this warp
product space is identified with the squared scale factor $a^2(t)$
in (1+3)-dimensional Robertson-Walker cosmology which describes the
time $t$ evolution of the 3-dimensional spatial hypersurface $M^{3|0}$ as
\begin{equation}
a(t)=\frac{L}{\sqrt{3}}\exp{\frac{t}{\sqrt{3}L}},
\end{equation}
which shows a de Sitter expansion of the universe in agreement with
the solution of FRW cosmological model with a cosmological constant
$L^{-2}$. 

Using the warp function \eqref{f(t)}, we obtain
$\Lambda_{_{(I)}}=0$. This result is in agreement with the fact that
the cosmological constant is an energy density corresponding to the
$tt$ component of Einstein equations and has no counterpart on the
supermanifold $M^{3|0}$, namely we have a Ricci flat
Einstein equation
\begin{equation}
G_{\alpha\beta} = 0.
\end{equation}

\subsection*{\textbf{ Case II }}
Let $\bar{M}=\textbf{I}^{0|1}\times_{f} M^{m|n}$  be a super
semi-Riemannian warped product.  Then
$\mathcal{O}_{\textbf{I}^{0|1}}= \mathbb{R}\oplus \mathbb{R}=a+b
\bar{t}$ ~ such that $a,b \in \mathbb{R}$
 and $(M,
 g)$ is  a super semi-Riemannian manifold an, $( x^i,x^\alpha)$ be an
even - odd  coordinate system on $M^{m|n}$, $\bar{t}$ be  a
coordinate on $\textbf{I}^{0|1}$ and $f\equiv f(\bar t)$. By
$|\bar{t}|,|i|,|\alpha|$ we mean the parities of
$\bar{t},x^i,x^\alpha$ respectively. So $|\bar{t}|= |\alpha|=1$ and
$|i|=0.$ we have
\begin{align}\label{wmetr2}
&\bar{g}_{\bar{t}\bar{t}}=\bar{g}(\partial_{\bar{t}},\partial_{\bar{t}})=0
,\nonumber\\
&\bar{g}_{\alpha \bar{t}}=\bar{g}_{\bar{t} \alpha }=\bar{g}_{\bar{t} i }=0 , \nonumber\\
&\bar{g}_{\alpha\beta}=fg_{\alpha\beta}, \nonumber\\
&\bar{g}_{ij}=fg_{ij},\nonumber\\
& \bar{g}_{i\alpha}=fg_{i\alpha}.
\end{align}
The Christoffel symbols $\Gamma_{IJ}^L$ satisfies  the following
equation
\begin{align}
\bar{\Gamma}_{IJ}^L \bar{g}_{_{LK}}
=&\frac{1}{2}\Big[\partial_{_I}\bar{g}_{_{JK}}-(-1)^{^{|K|(|I|+|J|)}}\partial_{_K}\bar{g}_{_{IJ}}+(-1)^{^{|I|(|J|+|K|)}}
\partial_{_J}\bar{g}_{_{KI}}\Big],
\end{align}
where $I,J, K,L$ stand for $\bar{t},i$ and $\alpha $. If at least
one of indices  $K$ , $I$ or $J$  is equal to $ \bar{t}$  then  the
both sides of equality are zero. So $\Gamma_{IJ}^{\bar{t}},
\Gamma_{\bar{t}\bar{t}}^{I}, \Gamma_{\bar{t}\bar{t}}^{\bar{t}},
\Gamma_{ \bar{t}I}^{\bar{t}}, \Gamma_{I\bar{t}}^{\bar{t}}  $ are
arbitrary elements. Therefor from now on for computation we consider
them to be zeros.\\
Similar to case I, we can obtain the following equations for Ricci
tensor
\begin{align}\label{wRicc2}
&\bar{\Ric}_{\bar{t}\bar{t}}=(n-m)\frac{f'^2}{4f^2} + \frac{f'^2}{2f^2},% ~~ \text{when}~~~ %m=n ~~~~ \text{or} ~~~
\nonumber\\
%\bar{\Ric}(\partial \bar{t} ,\partial \bar{t})
&\bar{\Ric}_{\bar{t}i}=\bar{\Ric}_{
\bar{t}\alpha}=0,\nonumber\\
&\bar{\Ric}_{\alpha\beta}={\Ric}_{
\alpha\beta},\nonumber\\
&\bar{\Ric}_{ij}={\Ric}_{ ij},\nonumber\\
&\bar{\Ric}_{\alpha i}={\Ric}_{\alpha i},
\end{align}
where $'$ denotes $\partial_{\bar t}$, and the scalar curvature
\begin{align}\label{wScalar2}
 \bar{S}= \frac{S_{_M}}{f}.
\end{align}
\begin{pro}
Let $(\bar{M},\bar{g})$ be super semi-Riemannian warped product as
above. If $(\bar{M},\bar{g})$  is an Einstein space, then
$\bar{\Ric}(\partial_{\bar{t}} ,\partial _{\bar{t}})=0$ and $m-n=
-2$.
\end{pro}
\begin{proof}
By Einstein equation with cosmological  constant
\begin{align}
\bar{G}=-\bar{\Lambda}\bar{g}.\label{6.1}
\end{align}
one has $\bar{G}_{\bar{t}\bar{t}}=
-\bar{\Lambda}\bar{g}_{\bar{t}\bar{t}}$. Since
$\bar{g}_{\bar{t}\bar{t}}=0$, thus $\bar{G}_{\bar{t}\bar{t}}=0$ and
$\bar{\Ric}_{\bar{t}\bar{t}}=0$. So  the first equality in
\eqref{wRicc2} shows $m-n=-2$.
\end{proof}
One can easily show that the Einstein gravitational tensor fields of
$(\bar{M},\bar{g})$ and $(M,g)$ are equal i.e. $\bar{G}=G$.
\begin{pro}
Let $(\bar{G}, \bar{\Lambda})$ be the Einstein gravitational tensor
field and cosmological constant of $(\bar{M},\bar{g})$, then we have
the following equations
\begin{align}
&\bar{\Lambda}=a \in \mathbb R,\\
&\bar{G}_{\alpha\beta}=G_{\alpha\beta}.
\end{align}
\end{pro}
The Einstein equations  on $(\bar{M},\bar{g})$ with cosmological
constant $\bar{\Lambda}$ induces the Einstein equations
 on $(M,
g)$, where the cosmological term $\Lambda$ is given by
\begin{equation}\label{27}
\Lambda(\bar t)= af(\bar t).
\end{equation}
We call $\bar{\Lambda}$ and ${\Lambda}$ here as
$\bar{\Lambda}_{_{(II)}}$ and ${\Lambda}_{_{(II)}}$, respectively.
It is seen that, unlike the case I, $\bar{\Lambda}_{_{(II)}}$ itself
is constant independent of $f(\bar t)$ and $(m, n)$. Therefore,
$f(\bar t)$ is no longer determined by the equation
$\bar{\Lambda}=a$ and this case describes a super warped product
spacetime with an arbitrary warp function $f(\bar t)$ which is a
result of identities
$\bar{g}_{\bar{t}\bar{t}}=\bar{\Ric}(\partial_{\bar{t}} ,\partial
_{\bar{t}})=\bar{G}_{\bar{t}\bar{t}}=0$. In other words, since there
is no information in the $\bar{t}\bar{t}$ component of Einstein
equations, the warp function is not determined. Using \eqref{27},
the cosmological term $\Lambda_{_{(II)}}$ is not determined too, and
so the Einstein equation
$G_{\alpha\beta}=-{\Lambda}_{_{(II)}}g_{\alpha\beta} $ is free of
physical information. Moreover, assuming a supermanifold $M^{3|0}$
and identifying the warp function with the squared scale factor
$a^2(t)$ in a (1+3)-dimensional Robertson-Walker cosmology, it turns
out that the $\bar t$ evolution of the supermanifold $M^{3|0}$ is
completely arbitrary and so this case does not describe a physically
viable cosmology.

%%%%%%%%%%%%%%%%
\subsection*{\textbf{ Case III }}
Let $(\bar{M} ,\bar{g})$  be a super semi-Riemannian warped product
where $\bar{M}=\textbf{I}^{1|1}\times_{f} M^{m|n}$  and $\bar{g}=
-dt^2 + dtd\bar{t} + fg_M$. In addition
$\textbf{I}^{1|1}=I^{1|0}\times I^{0|1}$ and define $f= h+\bar{t}k$
where $ h,k$ are smooth functions on $\textbf{I}$ and $h
> 0$.\\
Let $(t,\bar{t})$ be an even - odd coordinate system on
$\textbf{I}^{1|1}$,
 $( x^i,x^\alpha)$ be a coordinate system on $M^{m|n}$ and $f\equiv f(t, \bar t)$.
Then, all equalities \eqref{wmetr} and \eqref{wmetr2} are satisfied
simultaneously. Similar to case I and case II we can obtain the
following equations for Ricci tensor and scalar curvature
\begin{align}\label{wRicc4}
&\bar{\Ric}_{\bar{t}\bar{t}}=
 \frac{(m^2-n^2)}{2} \Big(\frac{m+n}{2}\frac{f_{\bar{t}}^2}{f^2}-\frac{f_{\bar{t}\bar{t}}}{f}\Big), \nonumber\\
&\bar{\Ric}_{tt}=\frac{(m^2-n^2)}{2}\left(-
 \frac{f_{tt}}{f} - \frac{(m+n-2)}{2}(\frac{f_t}{f})^2
 \right),
 \nonumber\\
&\bar{\Ric}_{t\bar{t}}=\frac{(m^2-n^2)}{2}\left(-\frac{f_{t\bar{t}}}{f}-\frac{(m+n-2)}
{2}\frac{f_{\bar{t}}f_t}{f^2}\right),
\nonumber\\
&\bar{\Ric}_{\bar{t}t}=
\frac{(m^2-n^2)}{2}\left(-\frac{f_{\bar{t}t}}{f}-\frac{(m+n-2)}{2}
\frac{f_{\bar{t}}f_t}{f^2}\right), \nonumber\\
 &\bar{\Ric}_{
\bar{t}\alpha}=\bar{\Ric}_ {t\alpha}= \bar{\Ric}_ {ti}= \bar{\Ric}_{
\bar{t}i}= 0,\nonumber\\
&\bar{\Ric}_{\alpha\beta}={\Ric}_{\alpha
\beta}-\frac{(m-n-1)(m+n)}{2}(\frac{f_tf_{\bar{t}}}{f})(g_{\alpha\beta})-
\frac{(m^2-n^2)}{4}(\frac{f^2_{\bar{t}}g_{\alpha\beta}}{f}),
\nonumber\\
&\bar{\Ric}_{i\alpha}={\Ric}_{i
\alpha}+ \frac{(m^2-n^2)}{4}(\frac{f^2_{\bar{t}}g_{i\alpha}}{f}), \nonumber\\
&\bar{\Ric}_{\alpha i}={\Ric}_{\alpha
i}+\frac{(m^2-n^2)}{4}(\frac{f^2_{\bar{t}}g_{\alpha i}}{f}), \nonumber \\
&\bar{\Ric}_{ij}={\Ric}_ {i
j}-\frac{(m-n-1)(m+n)}{2}(\frac{f_tf_{\bar{t}}}{f})(g_{ij}) -
\frac{(m+n)^2}{4}(\frac{f^2_{\bar{t}}g_{ij}}{f}),
\end{align}
\begin{equation}
 S=\dfrac{1}{f}S_M +\frac{(m+n)\Big(n(m-n-1)-m(m-1)\Big)}{2}\frac{f_tf_{\bar{t}}}{f^2}-
\frac{(m+n)}{2}
 \left(\frac{4n-m^2+n^2-2m}{2}\right)
 \frac{f^2_{\bar{t}}}{f^2} + (\frac{m^2-n^2-2}{2})\frac{f_{\bar{t}\bar{t}}}{f}.
\end{equation}
Let $(\bar{G}, \bar{\Lambda})$ be the Einstein gravitational tensor
field and cosmological constant of $(\bar{M},\bar{g})$ respectively,
then we have the following equations
\begin{align}
\bar{G}_{\bar{t}\bar{t}}=&(\frac{m^2 -
n^2}{2})\left(\frac{m+n}{2}\frac{f^2_{\bar{t}}}{f^2}-\frac{f_{\bar{t}\bar{t}}}{f}\right),\\
\begin{split}
\bar{G}_{tt}=&(\frac{m^2 -
n^2}{2})\left(-\frac{f_{tt}}{f}-(\frac{m+n-2}{2})\frac{f^2_t}{f^2}\right)+
\frac{S_M}{2f}+\frac{(m+n)(n(m-n-1)-m(m-1))}{4}\frac{f_tf_{\bar{t}}}{f^2}\\
\nonumber
&+\frac{(m+n)(4n-m^2+n^2-2m)}{4}\frac{f^2_{\bar{t}}}{2f^2}+(\frac{m^2-n^2-2}{4})\frac{f_{\bar{t}\bar{t}}}{f},
\end{split}\\
\begin{split}
\bar{G}_{t\bar{t}}=&-(\frac{m^2 -
n^2}{2})\frac{f_{t\bar{t}}}{f}-\frac{S_M}{2f}-\frac{(m+n)(4n-m^2+n^2-2m)}{4}\frac{f^2_{\bar{t}}}
{2f^2}+\frac{(m+n)(2n^2-mn+m-m)}{4}\frac{f_tf_{\bar{t}}}{f^2}\\
\nonumber  &-(\frac{m^2-n^2-2}{4})\frac{f_{\bar{t}\bar{t}}}{f},
\end{split}\\
\begin{split}
\bar{G}_{\alpha\beta}=&G_{\alpha\beta}-\frac{(m+n)}{2}\left(\frac{(m-n-1)(2+n)-m(m-1)}{2}
\right)\frac{f_tf_{\bar{t}}}{f}g_{\alpha\beta}-\frac{(m+n)(n^2-m^2+2n)}{4}\frac{f^2_{\bar{t}}}{2f}g_{\alpha\beta}
\\\nonumber  &-(\frac{m^2-n^2-2}{4})f_{\bar{t}\bar{t}}g_{\alpha\beta},
\end{split}\\
\begin{split}
\bar{G}_{\alpha i}=&G_{\alpha
i}-\frac{(m+n)}{2}\left(\frac{(n(m-n-1)-m(m-1))}{2}\right)\frac{f_tf_{\bar{t}}}{f}g_{\alpha
i}+\frac{(m+n)(4m+m^2-n^2-6n)}{4}\frac{f^2_{\bar{t}}}{2f}g_{\alpha
i}\\\nonumber  &-(\frac{m^2-n^2-2}{4})f_{\bar{t}\bar{t}}g_{\alpha
i},
\end{split}\\
\begin{split}
\bar{G}{ij}=&G_{ij}-\frac{(m+n)}{2}\left(\frac{(m-n-1)(2+n)-m(m-1)}{2}
\right)\frac{f_tf_{\bar{t}}}{f}g_{ij}-\frac{(m+n)(n^2-m^2+6n)}{4}\frac{f^2_{\bar{t}}}{2f}g_{ij}\\
\nonumber  &-(\frac{m^2-n^2-2}{4})f_{\bar{t}\bar{t}}g_{ij},
\end{split}
\end{align}
\begin{proof}
we use relation $\bar{G}=\bar{Ric}-1/2\bar{S}\bar{g}$ and similar to
Case I , II we obtain gravitational tensor field $\bar{G}$.
\end{proof}
\begin{pro}
Let $(\bar{G}, \bar{\Lambda})$ be the Einstein gravitational tensor
field and cosmological constant of $(\bar{M},\bar{g})$, then we have
the following equations
\begin{align}
\begin{split}
\bar{\Lambda}(t, \bar t)=&(1-\frac{n}{2})(\frac{m^2 -
n^2}{2})\left(-\frac{f_{tt}}{f}-(\frac{m+n-2}{2})\frac{f^2_t}{f^2}\right)+
\frac{(m+n)(2n(m-n-1)-m(m-1))}{4}\frac{f_tf_{\bar{t}}}{f^2}\\
\nonumber  &+\frac{m^2(n-2)+n^2(m-8)-(m^3+n^3)+10mn}{8}
\frac{f^2_{\bar{t}}}{f^2}+\frac{(m^2-n^2-2)}{4}\frac{f_{\bar{t}\bar{t}}}{f},\\
\end{split}\\
\begin{split}
G_{\alpha\beta}=&f g_{\alpha\beta}(1-\frac{n}{2})\\
\nonumber
&\times\left(\frac{(m+n)(m-n-1)}{2}\frac{f_tf_{\bar{t}}}{f^2}+
\frac{(m+n)(-6n+m^2-n^2+4m)}{8} \frac{f^2_{\bar{t}}}{f^2}-\frac{(m^2
-n^2)}{2}(-\frac{f_{tt}}{f}-(\frac{m+n-2}{2})\frac{f^2_t}{f^2})\right),
\end{split}
\end{align}
\end{pro}
\begin{proof}

The Einstein equations with cosmological  constant $\bar{\Lambda}$
is
\begin{align}
\bar{G}=-\bar{\Lambda}\bar{g}.\label{}
\end{align}
 Similar to case I , II  we
 obtain
$\bar{\Lambda}$ and $G_{\alpha\beta}$.
\end{proof}
The Einstein equations  on $(\bar{M},\bar{g})$ with cosmological
constant $\bar{\Lambda}$ induces the Einstein equations
 on $(M,
g)$, where the cosmological constant $\Lambda$ is given by
\begin{align}
\begin{split}
\Lambda(t, \bar t)=& - f(1-\frac{n}{2})\\
\nonumber
&\times\left(\frac{(m+n)(m-n-1)}{2}\frac{f_tf_{\bar{t}}}{f^2}+
\frac{(m+n)(-6n+m^2-n^2+4m)}{8} \frac{f^2_{\bar{t}}}{f^2}-\frac{(m^2
-n^2)}{2}(-\frac{f_{tt}}{f}-(\frac{m+n-2}{2})\frac{f^2_t}{f^2})\right).\\
\end{split}
\end{align}
We call $\bar{\Lambda}$ and ${\Lambda}$ here as
$\bar{\Lambda}_{_{(III)}}$ and ${\Lambda}_{_{(III)}}$, respectively.
The requirement $\bar{\Lambda}_{_{(III)}}={\it Const}$ can determine
the warp function $f(t, \bar t)$ and $\Lambda_{_{(III)}}(t, \bar
t)$, using which the Einstein equations
$\bar{G}_{AB}=-\bar{\Lambda}_{_{(III)}}\bar{g}_{AB}$ and
$G_{\alpha\beta}=-{\Lambda}_{_{(III)}}g_{\alpha\beta} $ are solved.
Solving $\bar{\Lambda}_{_{(III)}}={\it Const}$ as a nonlinear
quadratic partial differential equation for $f(t, \bar{t})$ is a
hard task. However, for instance, one can propose an ansatz
\begin{equation}
f(t, \bar{t})=e^{\alpha t+\beta \bar{t}},
\end{equation}
for which the equation $\bar{\Lambda}_{_{(III)}}={\it Const}$ reads
as the following algebraic equation
\begin{equation}\label{Const}
\begin{split}
&\alpha^2(1-\frac{n}{2})(\frac{n^2 -
m^2}{2})(\frac{m+n}{2})+\alpha\beta
\frac{(m+n)(2n(m-n-1)-m(m-1))}{4}\\
&+\beta^2\frac{m^2(n-2)+n^2(m-8)-(m^3+n^3)+2(m^2-n^2-2)+10mn}{8}=\mbox{Const}.
\end{split}
\end{equation}
For a given constant and given values of $m$ and $n$, the
coefficients $\alpha$ and $\beta$ satisfy the equation
\eqref{Const}. Moreover, $\Lambda_{_{(III)}}(t, \bar t)$ is also
obtained as follows
\begin{align}
\begin{split}
\Lambda_{_{(III)}}(t, \bar t)=& - e^{\alpha t+\beta \bar{t}}(1-\frac{n}{2})\\
&\times\left(\alpha^2\frac{(m^2
-n^2)}{2}(\frac{m+n}{2})+\alpha\beta\frac{(m+n)(m-n-1)}{2}+\beta^2
\frac{(m+n)(-6n+m^2-n^2+4m)}{8}\right).\\
\end{split}
\end{align}
Assuming a supermanifold $M^{3|0}$, the equation
$\bar{\Lambda}_{_{(III)}}={\it Const}$ and $\Lambda_{_{(III)}}(t,
\bar t)$ read as
\begin{equation}\label{35}
-\left(\frac{27}{4}\alpha^2+\frac{9}{2}\alpha\beta+\frac{31}{8}\beta^2\right)=\mbox{Const},
\end{equation}
\begin{equation}
\Lambda_{_{(III)}}(t, \bar t)=- e^{\alpha t+\beta
\bar{t}}\left(\frac{27}{4}\alpha^2+3\alpha\beta+\frac{63}{8}\beta^2\right).
\end{equation}
By choosing one of the coefficients $\alpha$ or $\beta$, the other
coefficient is fixed through the equation \eqref{35} and the $\bar{t}$ evolution of supermanifold $M^{3|0}$, namely $e^{\beta \bar{t}}$, determines its ${t}$ evolution $e^{\alpha{t}}$, and vice versa. A wide variety
of positive, negative, vanishing and imaginary values for $\alpha$ and $\beta$ is available. This is an interesting result in that we obtain a two
parameter $(\alpha, \beta)$ class of two-times cosmological models
for which the $t$ and $\bar{t}$ evolutions of the cosmological scale
factor $a(t, \bar{t})$ correspond to the same cosmological constant
$\bar{\Lambda}_{_{(III)}}$. Positive and/or negative values of $\alpha$ and $\beta$ may result in inflationary expanding or contracting (depending on the signs and values of $\alpha$ and $\beta$) supermanifold $M^{3|0}$. Vanishing values of both $\alpha$ and $\beta$ are possible if $\bar{\Lambda}_{_{(III)}}=0$
and result in a static state of supermanifold $M^{3|0}$, which corresponds
to Einstein static universe \cite{ESU}. Imaginary values
of $\alpha$ and/or $\beta$ may result in a full or partial wave-like evolution
of the supermanifold $M^{3|0}$ \cite{Wesson}.

\section*{Conclusion}
In this paper, we have studied three types of super warped product
space-times $\bar{M}_{_{(I)}}=\textbf{I}^{1|0}\times_f M^{m|n}$,
$\bar{M}_{_{(II)}}=\textbf{I}^{0|1}\times_{f} M^{m|n}$, and
$\bar{M}_{_{(III)}}=\textbf{I}^{1|1}\times_{f} M^{m|n}$, where
$M^{m|n}$ is a supermanifold of dimension $m|n$,
$\textbf{I}^{\delta|\delta'}$ is standard superdomain with
$\textbf{I}=(0,1)$ and $\delta,\delta' \in \{0,1\}$, subject to the
warp functions $f(t)$, $f(\bar t)$, and $f(t, \bar t)$,
respectively. By using the Generalized Robertson-Walker space-time,
as super warped product space, we have shown that the Einstein
equations $\bar{G}_{AB}=-\bar{\Lambda}\bar{g}_{AB}$, with
cosmological term $\bar{\Lambda}$ are reducible to the Einstein
equations $G_{\alpha\beta} = -\Lambda g_{\alpha\beta}$ on the
supermanifold $M$ with cosmological term ${\Lambda}$, where
$\bar{\Lambda}$ and ${\Lambda}$ are functions of $f(t)$, $f(\bar
t)$, and $f(t, \bar t)$, as well as ($m$, $n$). Then, by demanding
for constancy of $\bar{\Lambda}$ we have determined the warp
functions and $\Lambda$ which result in finding the solutions for
Einstein equations $\bar{G}_{AB}=-\bar{\Lambda}\bar{g}_{AB}$ and
$G_{\alpha\beta} = -\Lambda g_{\alpha\beta}$.

Interested in the cosmological solutions of Einstein equations with
cosmological constant in the special case of supermanifold $M^{3|0}$
(3-dimensional space) we have found, in each case of super warp product space, the
following results:
\begin{itemize}
\item In the first case, namely $\bar{M}_{_{(I)}}=\textbf{I}^{1|0}\times_f M^{3|0}$, the $t$ evolution of the supermanifold $M^{3|0}$ is described by
$t$ evolution of the cosmological scale factor as
$a(t)=\frac{L}{\sqrt{3}}\exp{\frac{t}{\sqrt{3}L}}$ which describes a
de Sitter expansion of the Universe with initial scale factor
$a(0)=\frac{L}{\sqrt{3}}$.
\item In the second case, namely $\bar{M}_{_{(II)}}=\textbf{I}^{0|1}\times_f M^{3|0}$, the $\bar t$ evolution of the supermanifold $M$ is completely arbitrary and so this case does not describe a physically viable cosmology.
\item In the third case, namely $\bar{M}_{_{(I)}}=\textbf{I}^{1|1}\times_f M^{3|0}$, the two-times evolution of supermanifold $M^{3|0}$ is described by $(t, \bar{t})$ evolution of the cosmological scale factor $a(t, \bar{t})=e^{\frac{\alpha t+\beta \bar{t}}{2}}$ which defines a two parameter $(\alpha, \beta)$ class of two-times cosmological models
corresponding to the same cosmological constant
$\bar{\Lambda}_{_{(III)}}$.
\end{itemize}

%\section*{Acknowledgements }
%The authors is very grateful Professor Mukut Mani Tripathi for some
%valuable suggestions to improve the presentation of this paper.
%\nocite{*}
%\bibliographystyle{acm}
%\bibliographystyle{siam}
%\bibliography{F:/bibtex/convolution,F:/bibtex/faghfouri}

\end{document}